\documentclass[11pt]{article}

\usepackage{times}  % Times font saves space

\usepackage{amsmath}
\usepackage{amsthm}
\usepackage{amssymb}
\usepackage{amscd}
\usepackage{fullpage}

\newcommand{\ignore}[1]{}

\ignore{
\setlength{\textwidth}{6.5 true in}  % Increased from 6.5-->6.528 (MJW)
\setlength{\oddsidemargin}{0. true in}
\setlength{\evensidemargin}{0. true in}
\setlength{\textheight}{9.0 true in}  % Increased from 8.5 - AS
\setlength{\topmargin}{-0.0 true in}   % Decreased from 0.0 - AS
\setlength{\parindent}{15pt}
\setlength{\parskip}{1pt}
}

\ignore{
\usepackage{atbeginend}
\BeforeBegin{lemma}{\vspace{-2mm} }
\BeforeBegin{theorem}{\vspace{-2mm} }
\BeforeBegin{definition}{\vspace{-2mm} }
\BeforeBegin{cor}{\vspace{-2mm} }
\AfterEnd{lemma}{\vspace{-2mm} }
\AfterEnd{theorem}{\vspace{-2mm} }
\AfterEnd{definition}{\vspace{-2mm} }
\AfterEnd{cor}{\vspace{-2mm} }
 }

\newtheorem{theorem}{Theorem}
\newtheorem{cor}[theorem]{Corollary}
\newtheorem{claim}[theorem]{Claim}
\newtheorem{lemma}[theorem]{Lemma}

\newtheorem{definition}[theorem]{Definition}

\renewcommand{\P}{\textup{P}}
\def\eps{{\epsilon}}

\usepackage{color}

\definecolor{Red}{rgb}{1,0,0}
\definecolor{Blue}{rgb}{0,0,1}
\definecolor{Olive}{rgb}{0.41,0.55,0.13}
\definecolor{Green}{rgb}{0,1,0}
\definecolor{MGreen}{rgb}{0,0.8,0}
\definecolor{DGreen}{rgb}{0,0.55,0}
\definecolor{Yellow}{rgb}{1,1,0}
\definecolor{Cyan}{rgb}{0,1,1}
\definecolor{Magenta}{rgb}{1,0,1}
\definecolor{Orange}{rgb}{1,.5,0}
\definecolor{Violet}{rgb}{.5,0,.5}
\definecolor{Purple}{rgb}{.75,0,.25}
\definecolor{Brown}{rgb}{.75,.5,.25}
\definecolor{Grey}{rgb}{.5,.5,.5}
\definecolor{Black}{rgb}{0,0,0}
\definecolor{Black}{rgb}{0,0,0}

\newcommand{\ve}{\varepsilon}

\newcommand{\ga}{\gamma}
\newcommand{\be}{\beta}
\newcommand{\tiO}{\tilde{O}}
\newcommand{\ov}[1]{\overline{#1}}
\newcommand{\si}{\sigma}

\begin{document}
\title{Noisy sorting without resampling}
\author{Mark Braverman\thanks{C.S. University of Toronto, partially 
supported by and NSERC CGS scholarship. Part of the work was done
while on a visit to IPAM, UCLA}
\and Elchanan Mossel
\thanks{Dept. of
Statistics, U.C. Berkeley. Supported by an Alfred Sloan fellowship
in Mathematics, by NSF grants DMS-0528488
and DMS-0548249 (CAREER) and by DOD ONR grant N0014-07-1-05-06. 
Part of this work was done while the author was visiting IPAM, UCLA}}
\date{\today}

\maketitle

\begin{abstract}
In this paper we study noisy sorting without re-sampling. 
In this problem there is an unknown order 
$a_{\pi(1)} < \ldots < a_{\pi(n)}$ where $\pi$ is a 
permutation on $n$ elements. 
The input is the status of $n \choose 2$ 
queries of the form $q(a_i,x_j)$, where $q(a_i,a_j) = +$  
with probability at least $1/2+\ga$  if $\pi(i) > \pi(j)$ 
for all pairs $i \neq j$, where $\ga > 0$ is a constant and 
$q(a_i,a_j) = -q(a_j,a_i)$ for all $i$ and $j$. It is assumed 
that the errors are independent. 
Given the status of the queries the goal is to find 
the maximum likelihood order. 
In other words, the goal is find a permutation $\sigma$ that minimizes 
the number of pairs $\sigma(i) > \sigma(j)$ where 
$q(\sigma(i),\sigma(j)) = -$. The problem so defined is 
the feedback arc set problem on 
distributions of inputs, each of which is a tournament 
obtained as a noisy perturbations of a linear 
order. Note that when $\ga < 1/2$ and $n$ is large, 
it is impossible to recover the original order $\pi$.

It is known that the weighted feedback are set problem on tournaments 
is NP-hard in general. Here we present an algorithm of running time 
$n^{O(\gamma^{-4})}$ and sampling complexity $O_{\gamma}(n \log n)$ that 
with high probability solves the noisy sorting without re-sampling problem. 
We also show that if $a_{\sigma(1)},a_{\sigma(2)},\ldots,a_{\sigma(n)}$ is an optimal 
solution of the problem then it is ``close'' to the original order. 
More formally, with high probability it holds that  
$\sum_i |\sigma(i) - \pi(i)| = \Theta(n)$ and 
$\max_i |\sigma(i) - \pi(i)| = \Theta(\log n)$. 

Our results are of interest in applications to ranking, such as ranking in sports, or ranking of search items based on comparisons by experts. 

 \end{abstract}

\newpage
\section{Introduction}
We study the problem of sorting in the presence of noise. 
While sorting linear orders is a classical well studied problem, 
the introduction of noise poses very interesting challenges. 
Noise has to be considered when ranking or sorting is applied in many 
real life scenarios. 

A natural example comes from sports. How do we rank a league of soccer teams 
based on the outcome of the games? It is natural to assume that there is a true 
underlying order of which team is better and that the games outcome represent 
noisy versions of the pairwise comparisons between teams. Note that in 
this problem it is impossible to ``re-sample'' the order between 
a pair of teams.  
As a second example, consider experts comparing various items  
according to their importance where each pair of elements is compared by one 
expert. 
It is natural to assume that the experts opinions 
represent a noisy view of the actual order of significance. The question is 
then how to  aggregate this information? 

\subsection{The Sorting Model}
We will consider the following probabilistic model of instances. 
There will be $n$ items denoted $a_1,\ldots,a_n$.  
There will be a {\em true order} given by a permutation $\pi$ on 
$n$ elements such that under the true order 
$a_{\pi(1)} < a_{\pi(2)} \ldots < a_{\pi(n-1)} < a_{\pi(n)}$. 
The algorithm will have access to $n \choose 2$ queries defined as follows.

\begin{definition} \label{def:noise}
For each pair $i,j$ the outcome of the comparison between $a_i$ and 
$a_j$ is denoted by $q(a_i,a_j) \in \pm$ where for all $i \neq j$ it holds 
that $q(a_i,a_j) = -q(a_j,a_i)$. We assume that the probability 
$q(a_i,a_j) = +$ is at least $p := \frac{1}{2} + \ga$ if $\pi(i) > \pi(j)$ 
and that the queries 
\[
\left\{ q(a_i,a_j) : 1 \leq i < j \leq n \right\}
\]
are independent conditioned on the true order. 
In other words, for any set 
\[
S = \{(i(1) < j(1)),\ldots,(i(k) < j(k))\},
\] 
any vector $s \in \{\pm\}^k$ and 
$(i < j) \notin S$ it holds that 
\begin{equation}
\label{cond:prob}
\P[q(a_i,a_j) = +  
| \forall 1 \leq \ell \leq k: q(a_{i(\ell)},a_{j(\ell)}) = s_{\ell}] = 
\P[q(a_i,a_j) = +].
\end{equation}
%{\bf Elchanan: It seems to me that it suffices to require that the conditional 
%probability is at least $p$, but not the same for all conditioning, no?\\}
It is further assumed that $1/2 < p = \frac{1}{2} + \ga < 1$.
\end{definition}

We will be interested in finding a ranking that will minimize the number of 
upsets. More formally:
\begin{definition}
Given $n \choose 2$ queries $q(a_i,a_j)$ the score $s_q(\sigma)$ 
of a ranking (permutation) $\sigma$ is given by 
\begin{equation} \label{eq:score}
s_q(\sigma) = \sum_{i, j : \sigma(i) > \sigma(j)} q(a_{\sigma(i)},a_{\sigma(j)}).
\end{equation}
We say that a ranking $\tau$ is {\em optimal} for $q$ if $\tau$ 
is a maximizer~(\ref{eq:score}) among all ranking. 

The {\em Noisy Sorting Without Resampling (NSWR)} 
problem is the problem of finding 
an optimal $\tau$ given $q$ assuming that $q$ is generated as in 
Definition~\ref{def:noise}.
\end{definition}
The problem of maximizing~(\ref{eq:score}) without any assumptions on the 
input distribution is called the {\em feedback arc set problem for tournaments} 
which is known to be NP-hard, see subsection~\ref{subsec:related} 
for references, more background and related models. 

The score~(\ref{eq:score}) has a clear statistical interpretation in the 
case where each query is answered correctly with probability $p$ exactly
%\footnote{Do we really need the exactness assumption here? M. In cases where different 
%comparisons have different $p$'s the corresponding $q$'s have 
%different weight in the maximum likelihood expression. E.}. 
In this case, for each permutation $\sigma$ we can calculate the probability 
$P[q | \sigma]$ of observing $q$ given that $\sigma$ is the true order. It is 
immediate to verify that $\log P[q | \sigma] = a s_q(\sigma) + b$ 
for two constants  $a > 0, b$. Thus in this case the optimal solution to the 
NSWR problem is identical with the {\em maximum likelihood} order that is 
consistent with $q$. This in particular implies that given a prior uniform 
distribution on the $n!$ rankings, any order $\sigma$ 
maximizing~(\ref{eq:score}) is also a maximizers of the posterior probability 
given $q$. So by analogy to problems in coding theory, 
see e.g.~\cite{Romann:97},  
$\sigma$ is a maximum likelihood decoding of the original order $\pi$.  
 
Note furthermore that one should not expect to be able to find the true order 
if $q$ is noisy. Indeed for any pair of adjacent elements we are only given 
one noisy bit to determine which of the two is bigger. 

\subsection{Related Sorting Models and Results}~\label{subsec:related} 
It is natural to consider the problem of finding an a ranking $\sigma$ that 
minimizes the score $s_q(\sigma)$ without making any assumptions on the input 
$q$. This problem, called the {\em feedback arc set problem for tournaments}  
is known to be NP hard~\cite{AiChNe:05,Alon:06}. 
However, it does admit PTAS~\cite{KenyonSchudy:07} achieving a $(1+\eps)$ 
approximation for 
\[
-\frac{1}{2}\left[s_q(\sigma)-{n \choose 2}\right].
\]
in time that is polynomial in $n$ and doubly exponential in $1/\eps$.
The results of~\cite{KenyonSchudy:07} are the latest in a long line of work 
starting in the 1960's and including~\cite{AiChNe:05,Alon:06}. 
See~\cite{KenyonSchudy:07} 
for a detailed history of the feedback arc set problem. 
 
A problem that is in a sense easier than NSWR is the problem where 
repetitions are allowed in querying. 
In this case it is easy to observe that the original 
order may be recovered in $O(n \log^2 n)$ queries with high probability. 
Indeed, one may perform any of the standard $O(n \log n)$ sorting 
algorithms and repeat each query $O(\log n)$ times in order to obtain the 
actual order between the queries elements with error probability $n^{-2}$ (say). 
More sophisticated methods allow to show that in fact the true order may be 
found in query complexity $O(n \log n)$ with high probability~\cite{FPRU:90}, 
see also~\cite{KarpKleinberg:07}. 

\subsection{Main Results}
In our main results we show that the NSWR problem is solvable in polynomial 
time with high probability and that any optimal order is close to the true 
order. More formally we show that

\begin{theorem} \label{thm:sort}
There exists a randomized algorithm that for any $\ga > 0$ and $\be > 0$ finds 
an optimal solution to the noisy sorting without resampling (NSWR) problem 
in time $n^{O((\beta + 1) \gamma^{-4})}$ except with probability $n^{-\be}$. 
\end{theorem}

\begin{theorem} \label{thm:dist}
Consider the NSWR problem and let $\pi$ be the {\em true} order and 
$\sigma$ be any optimal order than except with probability $O(n^{-\be})$ 
it holds that 
\begin{equation} \label{eq:lin_dist}
\sum_{i=1}^n |\sigma(i) - \pi(i)| = O(n), 
\end{equation}
\begin{equation} \label{eq:log_dev}
\max_i |\sigma(i) - \pi(i)| = O( \log n).
\end{equation}
\end{theorem}

Utilizing some of the techniques of~\cite{FPRU:90} 
it is possible to obtain the results of 
Theorem~\ref{thm:sort} with low sampling complexity. More formally, 
\begin{theorem} \label{thm:sampling}
There is an implementation of a sorting algorithm with the same guarantees 
as in Theorem~\ref{thm:sort} and whose sampling complexity is $C\, n \log n$ 
where $C = C(\beta,\gamma)$. 
\end{theorem}

It should be noted that the proofs can be modified to a more general case
where the conditional probability from \eqref{cond:prob} is always bounded 
from below by $p$ without necessarily being independent.

\subsection{Techniques}
In order to obtain a polynomial time algorithm for the NSWR problem is important to identify that any optimal solution to the problem is close to the true one. 
Thus the main step of the analysis is the proof of Theorem~\ref{thm:dist}.

To find efficient sorting we use an insertion algorithm. Given an optimal 
order on a subset of the items we show how to insert a new element. 
Since the optimal order both before and after the insertion of the element has 
to satisfy Theorem~\ref{thm:dist}, it is also the case that no element moves 
more than $O(\log n)$ after the insertion and re-sorting. 
Using this and a dynamic programing approach we derive an insertion algorithm 
in Section~\ref{sec:presorted}. The results of this section may be 
of independent interest in cases where it is known that a single element 
insertion into an optimal suborder cannot result in a new optimal order 
where some elements moved by much. 

The main task is to to prove Theorem~\ref{thm:dist} in 
Section~\ref{sec:dist}. We first prove~(\ref{eq:lin_dist}) by showing that 
for a large enough constant $c$, it 
is unlikely that any order $\sigma$ whose total distance is more than $c n$ will have $s_q(\sigma) \geq s_q(\pi)$, where $\pi$ is the original order.  
We then establish~(\ref{eq:log_dev}) in subsection~\ref{subsec:log_dev} 
using a bootstrap argument. 
The argument is based on the idea that if the discrepancy in the position of 
an element $a$ in an optimal order compared to the true order is more 
than $c \log n$ for a large constant $c$, then there must exist many elements 
that are ``close'' to $a$ that have also moved by much. 
This then leads to a contradiction with~(\ref{eq:lin_dist}). 

The final analysis of the insertion algorithm and the proof of Theorem~\ref{thm:sort} are provided in Section~\ref{sec:alg}. Section~\ref{sec:query} shows 
how using a variant of the sorting algorithm it is possible to achieve 
polynomial running time in sampling complexity $O(n \log n)$.

\subsection{Distances between rankings}
Here we define a few measures of distance between rankings that will be used 
later. First, given two permutations $\sigma$ and $\tau$ we define the 
{\em dislocation distance} by 
\[
d(\sigma,\tau) = \sum_{i=1}^n |\sigma(i) - \tau(i)|.
\]
Given a ranking $\pi$ we define $q_{\pi} \in \{\pm\}^{[n] \choose 2}$ 
so that $q_\pi(a_i,a_j) = +$ if
$\pi(i) > \pi(j)$ and $q_\pi (a_i,a_j) = -$ otherwise. 
Note that using this notation $q$ is obtained from $q_{\pi}$ by flipping each 
entry independently with probability $1-p = 1/2 - \gamma$. 
Given $q,q' \in \{\pm\}^{[n] \choose 2}$ 
we denote by 
\[
d(q,q') = \frac{1}{2} \sum_{i < j} |q(i,j) - q'(i,j)|
\] 
We will write $d(\sigma)$ for $d(\sigma,{id})$ where ${id}$ is the identity 
permutation and $d(q)$ for $d(q,q_{{id}})$. 
Below we will often use the following well known 
claim~\cite{DiaconisGraham:77}.
\begin{claim}
\label{cl:sd}
For any $\tau$,
$$
\frac{1}{2} d(\tau) \le d(q_{\tau}) \le   d(\tau).
$$
\end{claim}

%\begin{proof}
%Cosnider the Cayley graph of the symmetric group where two permutations are adjacent if they differ by an adjacent transpositions. We first claim that 
%$d(q_{\tau})$ is the distance of $\tau$ from the identity permutation 
%in this graph. This follows from the following two facts. First, each adjacent 
%transposition changes $d(q_{\tau})$ by exactly $1$. Second, if $\tau$ is not the identity permutation, then there exists $i,i+1$ such that 
%$\tau(i+1) < \tau(i)$.  Exchanging $i,i+1$ will then reduce $d(q_{\tau})$ by 
%exactly $1$. 
%Next we note that each adjacent transposition changes $d(\tau)$ by at most $2$. 
%This shows that $d(\tau) \leq 2 d(q_{\tau})$.

%Finally, note that if $\tau$ is 
%not the identity permutation then there exists a pair $i, i+1$ such that 
%$\tau(i)>i$ and $\tau(i+1) \le i+1 < \tau(i)$. Exchanging them will reduce $d(q_{\tau})$ by exactly $1$ 
%and $d(\tau)$ by either $1$ or $2$. {\bf actually i don't think it is true -- it may
%stay constant, needs a fix}
%\end{proof}

\section{Sorting a presorted list} \label{sec:presorted}

In this section we prove that if a list is pre-sorted so that each 
element is at most $k$ positions away from its location in the optimal
ordering, then the optimal sorting can be found in time $O(n^2 \cdot 2^{6 k})$.

\begin{lemma}
\label{lem:presort}
Let $a_1$, $a_2$, $\ldots$,  $a_n$ be $n$ elements 
together with noisy queries $q$. Suppose that 
we are given that there is an optimal ordering 
$a_{\si(1)},a_{\si(2)},\ldots,a_{\si(n)}$, such
that $|\si(i)-i|\le k$ for all $i$.  
Then we can find such an optimal $\si$ in time $O(n^2\cdot 2^{6k})$. 
\end{lemma}
In the applications below $k$ will be $O(\log n)$.  
Note that a brute force search over all possible $\si$ would require time 
$k^{\Theta(n)}$. Instead we use dynamic programing to reduce the running time. 

\begin{proof}

We use a dynamic programming technique to find an optimal sorting. 
In order to simplify notation we assume that the true ranking $\pi$ is the 
identity ranking. In other words, $a_1 < a_2 \ldots < a_n$. 
Let $i<j$ be any 
indices, then by the assumption, the elements in the optimally ordered interval 
$$I=[a_{\si(i)}, a_{\si(i+1)},\ldots, a_{\si(j)}]$$ satisfy 
$I^{-} \subset I \subset I^{+}$ where
$$
I^+ = [a_{i-k},a_{i-k+1},\ldots,a_{j+k}],\quad, 
I^- = [a_{i+k},a_{i+k+1},\ldots,a_{j-k}].
$$
Hence selecting the set $S_I=\{a_{\si(i)}, a_{\si(i+1)},\ldots, a_{\si(j)}\}$ involves 
choosing a set of size $j-i+1$ that contains the elements of $I^-$ and is contained in
$I^+$. This involves selecting $2k$ elements from the list (or from a subset of the list)
$$
\{a_{i-k}, a_{i-k+1},\ldots, a_{i+k-1}, a_{j-k+1}, a_{j-k+2},\ldots, a_{j-k}\}
$$
which has $4k$ elements. Thus the number of such $S_I$'s is bounded by $2^{4k}$. 

We may assume without loss of generality that $n$ is an exact power of $2$.  
Denote by 
$I_0$ the interval containing all the elements. Denote by $I_1$ the left
half of $I_0$ and by $I_2$ its right half. Denote by $I_3$ the left half
of $I_1$ and so on. In total, we will have $n-1$ intervals of lengths $2,4,8,\ldots$. 

For each $I_t=[a_i,\ldots,a_j]$ let $S_t$ denote the possible ($<2^{4k}$) sets of 
the elements $I'_t=[a_{\si(i)},\ldots,a_{\si(j)}]$. We use dynamic programming to store an
optimal ordering of each such $I'_t \in S_t$. The total number of $I'_t$'s we will 
have to consider is bounded by $n \cdot 2^{4k}$. We proceed from $t=n-1$ down
to $t=0$ producing and storing an optimal sort for each possible $I'_t$. For 
$t=n-1,n-2,\ldots,n/2$ the length of each $I'_t$ is $2$, and the optimal sort 
can be found in $O(1)$ steps. 

Now let $t<n/2$. We are trying to find an optimal sort of a given $I'_t=[i,i+2 s-1]$. We do 
this by dividing the optimal sort into two halves $I_l$ and $I_r$ and trying to 
sort them separately. We know that $I_l$ must contain all the elements in
$I'_t$ that come from the interval $[a_{1},\ldots,a_{i+s-1-k}]$ and must be contained
in the interval $[a_{1},\ldots,a_{i+s-1+k}]$. Thus there are at most $2^{2k}$ choices for 
the elements of $I_l$, and the choice of $I_l$ determines $I_r$ uniquely. For 
each such choice we look up an optimum solution for $I_l$ and for $I_r$ in the
dynamic programming table. 
Among all possible choices of $I_l$ we pick the best one. This is done 
by recomputing the score $s_q$ for the joined interval, 
and takes at most $|I'_t|^2$ time. Thus the total 
cost will be 
$$
\sum_{i=1}^{\log n} \#\mbox{intervals of length $2^i$}\cdot \#\mbox{checks} \cdot \mbox{cost of check} = 
\sum_{i=1}^{\log n} O\left(\frac{n\cdot 2^{4k}}{2^i} \cdot 2^{2k}\cdot 2^{2i} \right)=
O(n^2 \cdot 2^{6k}). 
$$
\end{proof}

\section{The Discrepancy between the true order and Optima} 
\label{sec:dist}

The goal of this section is to establish that with high probability any optimum
solution will not be far from the original solution. We first establish that 
the orders are close on average, and then that they are pointwise close to each 
other. 

\subsection{Average proximity} \label{subsec:lin_dist}

We prove that with high probability, the total difference between the original and any optimal ordering is linear in the length of the interval.

We begin by bounding the probability that a specific permutation $\sigma$ 
will beat the original ordering. 

\begin{lemma} \label{lem:two_perms}
Suppose that the original ordering is $a_1 < a_2 \ldots < a_n$. 
Let $\sigma$ be another permutation. Then the probability that 
$\sigma$ beats the identity permutation is bounded from above by
\[
P[Bin(d(q_{\sigma}),1/2+\ga) \leq d(q_{\sigma})/2] \leq \exp(-2 d(q_{\sigma}) \ga^2)
\]
\end{lemma}

\begin{proof}
In order for $\sigma$ to beat the identity, it needs to beat it in at least 
half of the $d(q_{\sigma})$ pairwise relation where they differ. This proves 
that the probability that it beats the identity is exactly 
$P[Bin(d(q_{\sigma}),1/2+\ga) \leq d(q_{\sigma})/2]$. The last inequality follows by a 
Chernoff bound. 
\end{proof}

\begin{lemma} \label{lem:distance_distribution}
The number of permutations $\tau$ on $[n]$ satisfying $d(\tau) \leq c\,n$ 
is at most
\[
  2^n\, 2^{(1+c)\, n \,H(1/(1+c))}.
\]
\end{lemma}

Here $H(x)$ is the binary entropy of $x$ defined by 
$$
H(x) = - x\log_2 x - (1-x) \log_2(1-x) < -2 x \log_2 x,
$$
for small $x$. 

\begin{proof}
Note
that each $\tau$ can be uniquely specified by the values of 
$s(i)=\tau(i)-i$, that we are given that $\sum |s(i)|$ is exactly 
$d(\tau) \le c n$. Thus there is an injection of $\tau$'s with $d(\tau)=m$ into 
sequences of $n$ numbers which in absolute values add up to $m$. 
It thus suffices to bound the number of such sequences.
The number of unsigned sequences equals the number of 
ways of placing $m$ balls in $n$ bins, which is equal to ${n+m-1}\choose{n-1}$.
Signs multiply the possibilities by at most $2^n$. Hence the total number of 
$\tau$'s with $d(\tau)=m$ is bounded by $2^n \cdot {{n+m-1}\choose{n-1}}$. Summing 
up over the possible values of $m$ we obtain
\begin{equation}
\sum_{m=0}^{c n} 2^n \cdot {{n+m-1}\choose{n-1}} < 2^n \cdot {{n+ c n}\choose{n}} 
\leq 2^n\, 2^{(n+ c n)\,H(n/(n+c n))}. 
\end{equation}
\end{proof}

\begin{lemma}\label{lem:p0}
Suppose that the true ordering is $a_1 < \ldots < a_n$ and $n$ is large enough. 
Then if $c \geq 1$ and 
\[
\ga^2 c > 1 + (1+c) H(1/(1+c)),
\]
the probability that any ranking $\sigma$ is optimal and 
$d(\sigma) > c n$ is at most $\exp(-c n \ga^2/10)$ for sufficiently large $n$.
In particular, as $\ga \to 0$, it suffices to take 
\[
c = O(-\ga^{-2} \log \ga) = \tiO (\ga^{-2}).
\]
\end{lemma}

\begin{proof}
Let $\sigma$ be an ordering with $d(\sigma) > c n$. Then by Claim~\ref{cl:sd} 
we have $d(q_{\sigma}) > c n /2$. Therefore the probability that such 
an ordering will beat the identity is bounded by 
$\exp(-c n \ga^2)$ by Lemma~\ref{lem:two_perms}. 
We now use union bound and Lemma~\ref{lem:distance_distribution} to obtain 
the desired result. 
\end{proof}

\subsection{Pointwise proximity} \label{subsec:log_dev}

In the previous section we have seen that it is unlikely that the {\em average} element 
in the optimal order is more than a constant number of positions away from its original
location. Our next goal is to show that the {\em maximum} dislocation of an element 
is bounded by $O(\log n)$. As a first step, we show that one ``big" dislocation
is likely to entail many ``big" dislocations.

\begin{lemma}
\label{lem:p1}
Suppose that the true ordering of $a_1,\ldots,a_n$ is given by the identity ranking, i.e., $a_1 < a_2 \ldots < a_n$. Let $1 \leq i < j \leq n$ 
be two indices and $m=j-i$. 
Let $A_{ij}$ be the event that there is an optimum ordering $\si$ such 
that $\si(i)=j$ and 
\[
(\si[1,i-\ell-1] \cup \si[j+\ell+1,n]) \cap [i,j-1] \leq \ell,
\]
i.e., at most $\ell$ elements are mapped to the interval $[i,j-1]$ from 
outside the interval $[i-\ell,j+\ell]$ by $\si$, where $\ell=\left\lfloor \frac{1}{6}\ga m\right\rfloor$. Then
$$
P(A_{ij})<p_1^m,
$$
where $p_1 = \exp(-\ga^2/16) <1$.
\end{lemma}

\begin{proof}
The assumption that $\si$ is optimal implies in particular that moving the
$i$-th element from the $j$-th position where it is mapped by $\si$ back to the 
$i$-th position does not improve the solution. The event $A_{ij}$ implies that 
among the elements $a_k$ for $k\in [i-\ell, j+\ell]$ at least $m/2-\ell$ 
satisfy $q(k,i) = -$. This means that at least 
\[
\frac{m}{2}-2 \ell-1 > 
\frac{m}{2} - \frac{\ga}{2} m + \frac{\ell}{2} > 
\left(\frac{1}{2}-\frac{\ga}{2}\right)(m+\ell)
\]
of the elements $a_k$ for $k\in [i+1,j+\ell]$ must satisfy $q(k,i) = -$. 
The probability of this occurring is less than
$$
\exp\left(\frac{-\frac{m+\ell}{2} \, (\ga/2)^2}{2}\right) = p_1^{m+\ell}
$$
using Chernoff bounds.
%where $p_1 = \sqrt{(1-2\ga)(1+2\ga)}<1$.
%{\bf need to be checked. maybe need $2 \ell$ power above.}
\end{proof}

As a corollary to Lemma \ref{lem:p1} we obtain the following using a simple
union-bound. For the rest of the proof all the $\log$'s are base $2$. 

\begin{cor}
\label{cor:p1}
Let 
$$m_1=(-\log \ve +2\log n/\log(1/p_1)) = O((-\log \ve+ \log n)/ \ga^2),$$
 then $A_{ij}$ does not occur for any $i,j$ with
$|i-j|\ge m_1$ with probability $>1-\ve$.
\end{cor}

Next, we formulate a corollary to Lemma \ref{lem:p0}. 

\begin{cor}
\label{cor:p0}
 Suppose that $a_1<a_2<\ldots<a_n$ is 
the true ordering. Set $m_2=2 m_1$. For each interval $I=[a_i,\ldots,a_j]$ with
at least $m_2$ elements consider 
all the sets $S_I$  which   contain the elements from
$$
I^- = [a_{i+m_2},\ldots,a_{j-m_2}],
$$
and are contained in the interval 
$$
I^+ = [a_{i-m_2},\ldots,a_{j+m_2}].
$$
Then with probability $>1-\ve$ all such sets $S_I$ do not have an optimal ordering that has 
a total deviation
from the true of more than $c_2\, |i-j|$, with 
$$c_2 = \frac{70}{\ga^2} = O(\ga^{-2}),$$
 a constant.
\end{cor}

\begin{proof}
There are at most $n^2 \cdot 2^{4 m_2}$ such intervals. The probability of each interval
not satisfying the conclusion is bounded by Lemma \ref{lem:p0} with 
$$
e^{-c_2 m_2 \ga^2/10} = e^{-7 m_2} < 2^{-7 m_2} = 
 2^{-m_2}\cdot 2^{-2 m_2} \cdot 2^{-4 m_2} <
\ve \cdot n^{-2} \cdot 2^{-4 m_2}.
$$
The last inequality holds because $m_2>\max(\log n,-\log \ve)$.
By taking a union bound over all the sets we obtain the statement of the corollary. 
\end{proof}

We are now ready to prove the main result on the pointwise distance between an optimal 
ordering and the original. 

\begin{lemma}
\label{lem:pmain}
Assuming that the events from Corollaries \ref{cor:p1} and \ref{cor:p0} hold, 
if follows that for each optimal ordering $\si$ and for each $i$, $|i-\si(i)|< c_3\log n$, 
where 
$$c_3 = 500\, \ga^{-2}\cdot\frac{m_2}{\log n}=O(\ga^{-4} (-\log \ve/\log n +1))$$
 is a constant.
 In particular, this conclusion holds with probability $>1-2 \ve$.
\end{lemma}

\begin{proof}
Assume that the events from both corollaries hold, and let $\si$ be an optimal
ordering.  We say that a position $i$ is {\em good} if 
there is no index $j$ such that $\si(j)$ is on the other side of $i$ from $j$ and 
$|\si(j)-j|\ge m_2$. In other words, $i$ is good if there is 
no ''long'' jump over $i$ in $\si$. 
In the case when $i=j$ or $i=\si(j)$ for a long
jump, it is not considered good. An index that is not good is bad. 
An interval $I$ is 
bad if all of its indices are bad. Our goal is to show that there are no bad intervals of length $\ge c_3 \log n$. 
This would prove the lemma, since if there is an $i$ with 
$|i-\sigma(i)| >   c_3 \log n$ then there is a bad interval of length at least 
$c_3 \log n$.

Assume, for contradiction, that $I=[i,\ldots,{i+t-1}]$ is a bad interval 
of length $t\ge c_3 \log n$, such that $i-1$ and $i+t$ are both good (or lie
beyond the endpoints of $[1,\ldots,n]$). Denote by $S$ the set of elements 
that is mapped to $I$ by $\si$. Denote the  indices in $S$ in their original 
order by $i_1<i_2<\ldots<i_t$, i.e., we have: 
$\{\sigma(i_1),\ldots,\sigma(i_t)\} = I$. 

By the goodness of the endpoints of $I$ we have
\[
[i+m_2, i+t-1-m_2] \subset \{ i_1,\ldots,i_t \} \subset [i-m_2, i+t-1+m_2].
\] 
Denote the permutation induced by $\si$ on $S$ by $\si'$ so 
$\sigma(i_j) < \sigma(i_{j'})$ is equivalent to $\si'(j) < \si'(j')$.
 The permutation 
$\si'$ is optimal, for otherwise it would have been possible to improve $\si$
by improving $\si'$.

By Corollary \ref{cor:p0} and Claim \ref{cl:sd}, we have 
\[
d(q_{\si'})\le d(\si')\le c_2 t.
\]
In how many switches can
the elements of $S$ participate under $\si$? They participate in switches with other
elements of $S$ to a total of $d(q_{\si'})$. In addition, they participate in switches
with elements that are not in $S$. These elements must originate at the margins of the 
interval $i$: either in the interval $[i-m_2,i+m_2]$ or the interval $[i+t-1-m_2,i+t-1+m_2]$.
Thus, each contributes at most $2 m_2$ switches with elements of $S$. There are at most $2 m_2$ 
such elements. Hence the total number of switches between elements in $S$ and in $\ov{S}$ is at most
$4 m_2^2$. Hence
\begin{equation}
\label{eq1}
\sum_{i\in S} |\si(i)-i| \le \sum_{i\in S} \#\{\mbox{switches $i$ participates in}\} \le
4 m_2^2 + 2 d(q_{\si'}) \le 4 m_2^2 + 2 c_2 t.
\end{equation}

We assumed that the entire interval $I$ is bad, hence for every position $i$ there is 
an index $j_i$ such that $|\si(j_i)-j_i|\ge m_2$ and such that $i$ is in the interval
$J_i=[j_i,\si(j_i)]$ (or the interval $[\si(j_i),j_i]$, depending on the order). 
Consider all such $J_i$'s. By a Vitali covering lemma argument we can 
choose a disjoint collection of them   whose total length is at 
least $|I|/3$. The argument proceeds as follows: Order the intervals  
in a decreasing length order (break ties arbitrarily). Go through the list and 
add a $J_i$ to our collection if it is disjoint from all the currently selected intervals.
We obtain a collection $J_1, \ldots, J_k$ of disjoint intervals of the for $[j_i,\si(j_i)]$. 
Denote the length of the $i$-th interval by $t_i = |j_i-\si(j_i)|$. 
Let $J_i'$ be the ''tripling" of the interval $J_i$: $J_i'=[j_i-t_i,\si(j_i)+t_i]$. 
We claim that the $J_i'$-s cover the entire interval $I$. Let $m$ be a position on the interval 
$I$. Then there is an interval of the form $[j,\si(j)]$ (or $[\si(j),j]$) that covers
$m$. Choose the longest such interval $J'=[j,\si(j)]$. If $J'$ has been selected to 
our collection then we are done. If not, it means that $J'$ intersects a longer interval 
$J_i$ that has been selected. This means that $J'$ is covered by the tripled interval $J_i'$. 
In particular, $m$ is covered by $J_i'$. We conclude that 
$$
t= \mbox{length}(I) \le \sum_{i=1}^{k} \mbox{length}(J_i') = 3 \sum_{i=1}^{k} t_i. 
$$
Thus $\sum_{i=1}^{k} t_i \ge t/3$. This concludes the covering argument.

We now apply Corollary \ref{cor:p1} to the intervals $J_i$. We conclude that on an interval
$J_i$ the contribution of the elements of $S$ that are mapped to $J_i$ to the sum of 
deviations under $\si$ is at least $\ell_i^2$ where $\ell_i = \frac{1}{6}\ga t_i$.
Thus 
\begin{multline*}
\sum_{i\in S} |\si(i)-i| \ge \sum_{j=1}^{k} \ell_j^2 = \frac{1}{36}\ga^2 \cdot \sum_{j=1}^{k} t_j^2
\ge \frac{1}{36}\ga^2 \cdot m_2 \cdot \sum_{j=1}^{k} t_j \\
\ge \frac{1}{36}\ga^2 \cdot m_2 \cdot t/3 \ge m_2 \cdot \frac{1}{125}\ga^2  \cdot c_3 \log n +
\frac{1}{800}\ga^2  \cdot m_2 t \\
> m_2 \cdot (4 m_2)+ 2 c_2 t =4 m_2^2 + 2 c_2 t,
\end{multline*}
for sufficiently large $n$.
The result contradicts \eqref{eq1} above.
Hence there are no bad intervals of length $\ge c_3 \log n$, which completes the proof. 
\end{proof}

%{\bf It would be nice if all c's are explicit}.

\section{The algorithm} \label{sec:alg}

We are now ready to give an algorithm for computing the optimal ordering with high
probability in polynomial time. Note that Lemma \ref{lem:pmain} holds for any interval
of length $\le n$ (not just length exactly $n$). Set $\ve = n^{-\be-1}/4$. 
Given an input, let  
$S\subset \{a_1,\ldots,a_n\}$ be a random set of size $k$. The probability that there is an 
optimal ordering $\si$ of $S$ and an index $i$ such that $|i-\si(i)|\ge c_3 \log n$, where
$$
c_3 = O(\ga^{-4} (-\log \ve/\log n +1)) =O(\ga^{-4} (\be +1)),
$$
is bounded by $2\ve$ by  Lemma \ref{lem:pmain}. Let 
$$S_1\subset S_2\subset\ldots\subset S_n$$
be a randomly selected chain of sets such that $|S_k|=k$. Then the probability that 
an element of an optimal order of any of the  $S_k$'s deviates from its original 
location by more than $c_3 \log n$    is bounded by 
$2 n \ve =  n^{-\be}/2$. We obtain: 

\begin{lemma}
\label{lem:chain}
Let $S_1\subset\ldots\subset S_n$ be a chain of randomly chosen subsets with $|S_k|=k$. 
Denote by $\si_k$ an optimal ordering on $S_k$. 
Then with probability $\ge 1-n^{-\be}/2$,
for each  $\si_k$ and for each $i$, $|i-\si_k(i)|<c_3\log n$, where 
$c_3 = O(\ga^{-4}(\be+1))$ is a constant.
\end{lemma}

We are now ready to prove the main result.

\begin{theorem}
\label{thm:main}
There is an algorithm that runs in time $n^{c_4}$ where 
$$c_4= O(\ga^{-4}(\be+1))$$ is a constant
that outputs an optimal ordering with probability $\ge 1-n^{-\be}$.
\end{theorem}

\begin{proof} 
First, we choose a random chain of sets $S_1\subset\ldots\subset S_n$ such 
that $|S_k|=k$. Then by Lemma \ref{lem:chain}, with probability $1- n^{-\be}/2$,
 for each optimal order $\si_k$ of $S_k$ and for each $i$, $|i-\si_k(i)|<c_3\log n$.
 We will find the orders $\si_k$ iteratively until we reach $\si_n$ which will be 
 an optimal order for our problem. Denote $\{a_k\}=S_k-S_{k-1}$. Suppose that we have 
 computed $\si_{k-1}$ and we would like to compute $\si_k$. We first  insert
 $a_k$ into a location that is close to its original location as follows.
Break $S_k$ into blocks $B_1, B_2,\ldots, B_s$ of length $c_3 \log n$. We claim 
that with probability $>n^{-\be-1}/2$ we can pinpoint the block $a_k$ belongs to 
within an error of $\pm 2$, thus locating $a_k$ within $3 c_3 \log n$ of its original 
location. 

Suppose that $a_k$ should belong to block $B_i$. Then by our assumption on 
$\si_{k-1}$, $a_k$ is bigger than any element in $B_1,\ldots,B_{i-2}$ and smaller
than any element in $B_{i+2},\ldots,B_s$. By comparing $a_k$ to each element in 
the block and taking majority, we see that the probability of having an incorrect
comparison result with a block $B_j$ is bounded by $n^{-\be-2}/2$. Hence 
the probability that $a_k$ will not be placed correctly up to an error of two 
blocks is bounded by $n^{-\be-1}/2$ using union bound.

% We claim 
%  the element $a_k$ lands at most $c_5 \log n$ positions away from 
% where it was in the original order with probability at least $1- n^{-\be-2}$. Here 
% we can take 
% $$
% c_5 = \frac{2}{-\log((1+2\ga)(1-2\ga))} \cdot \left( 
% \be+3 - 2 c_3 \log (1-2 \ga)
% \right)  = O(\ga^{-5}\be).
% $$
% Suppose that $a_k$ is $s>c_5\log n$ away from its original location. Suppose that its location
% $j$ is to the right from its original location $j-s$. Then at least $s/2-2 c_3\log n$ of the elements
% originally between the locations $j-s+c_3\log n$ and $j-c_3 \log n$ must have an inverted 
% relation with $a_k$. By our choise of $c_5$, this will happen with probability $<1-n^{-\be-3}$.
%  There are $<n$ possible ''far" location where $a_k$ could have landed and the probability 
%  of each is $<n^{-\be-3}$. By union bound, the probability of landing farther than $c_5\log n$ from 
%  the original location is bounded by $n^{-\be-2}$. 
  
  Hence after inserting $a_k$ we obtain an ordering of $S_k$ in which each element is at most 
  $3 c_3 \log n$ positions away from its original location. Hence each element is at most $4 c_3 \log n$ 
  positions away from its optimal location in $\si_k$. Thus, by Lemma \ref{lem:presort} we can obtain 
  $\si_k$ in time $O(n^{24 c_3+2})$. The process is then repeated. 
  
  The probability of each stage failing is bounded by $n^{-\be-1}/2$. Hence the probability of the 
  algorithm failing assuming the chain  $S_1\subset\ldots\subset S_n$ satisfies Lemma \ref{lem:chain}
  is bounded by $n^{-\be}/2$. 
  Thus the algorithm runs in time $O(n^{24 c_3 +3})$ and has a failure probability of 
  at most 
 $  n^{-\be}/2 +  n^{-\be}/2 =  n^{-\be}.$
\end{proof}

\section{Query Complexity} \label{sec:query}
Here we briefly sketch the proof of Theorem~\ref{thm:sampling}. Recall that the 
theorem states that  
although the running time of the algorithm is a polynomial of $n$ whose degree 
depends on $p$, the query complexity of a variant of the algorithm is 
$O(n \log n)$. Note that there are two types of queries. The first type 
is comparing elements in the dynamic programing, while the second is when inserting new elements. 

\begin{lemma} \label{lem:dynamic_query}
For all $\beta > 0, \gamma < 1/2$ there exists $c(\beta,\gamma) < \infty$ 
such that the total number of comparisons performed in the dynamic 
programing stage is $O(n \log n)$ of the algorithm is at most 
$c\, n \log n$ except with probability $O(n^{-\beta})$.
\end{lemma}

\begin{proof}
Recall that in the dynamic programing stage, each element is compared with 
elements that are at current distance at most $c_0 \log n$ from it where 
$c_0 = c_0(\beta,\gamma)$. 

Consider a random insertion order of the elements $a_1,\ldots,a_n$. 
Let $S_{n/2}$ denote the set of elements inserted up to the $n/2$ insertion. 
Then by standard concentration results it follows that there exists 
$c_1(c_0,\beta)$ such that for all $1 \leq i \leq n - c_1 \log n$ it holds that 
\begin{equation} \label{eq:up_int_query}
|[a_i,a_i+c_1 \log n] \cap S_{n/2}| \geq c_0 \log n,
\end{equation}
and for all $c_1 \log n \leq i \leq n$ it holds that 
\begin{equation} \label{eq:low_int_query}
|[a_i-c_1 \log n,a_i] \cap S_{n/2}| \geq c_0 \log n
\end{equation}
except with probability at most $n^{-\beta-1}$. 
Note that when~(\ref{eq:up_int_query}) and~(\ref{eq:low_int_query}) both hold 
the number of different queries used in the dynamic programing while inserting the 
elements in $\{ a_1,\ldots,a_n \} \setminus S_{n/2}$ is at most $2 c_1 n \log n$. 

Repeating the argument above for the insertions performed from 
$S_{n/4}$ to $S_{n/2}$, from $S_{n/8}$ to $S_{n/4}$ etc. we obtain that the total 
number of queries used is bounded by:
\[
2 c_1 \log n (n + n/2 + \ldots + 1) \leq 4 c_1 n \log n,
\]
except with probability $2 n^{-\beta}$. This concludes the proof. 
\end{proof}

Next we show that there is implementation of insertion that requires only 
$O(\log n)$ comparisons per insertion.
\begin{lemma} \label{lem:insertion_query}
For all $\beta > 0$ and $\gamma < 1/2$ there exists a $C(\beta,\gamma)=O(\ga^{-2}(\be+1))$ 
and $c(\beta,\gamma)=O(\ga^{-4}(\be+1))$ 
such that except with probability $O(n^{-\beta})$ it is 
 possible to perform the insertion in the proof of Theorem \ref{thm:main} 
so that each element is inserted using at most $C \log n$ 
comparisons, $O(\log n)$ time
and the element 
is placed a distance of at most $c \log n$ from its optimal location. 
\end{lemma}
%Note that the lemma implies that the optimal insertion location can be found 
%using $(C + 2 c) \log n$ comparisons and $O(\log n)$ time. 

\begin{proof}
Bellow we assume (as in the proof of Theorem \ref{thm:main}) that there exists 
$c_1(\beta,\gamma)=O(\ga^{-4}(\be+1))$ 
such that 
at all stages 
of the insertion and for each item, the distance between the location of the 
item in the original order and the optimal order is at most $c_1 \log n$. 
This will result in an error with probability at most $n^{-\beta}/2$. 
Let $k=k(\gamma)=O(\ga^{-2})$ be a constant such that 
\[
P[Bin(k,1/2 + \gamma) > k/2] > 1 - 10^{-3}.
\]

Let $c_2=O(\be+1)$ be chosen so that 
\begin{equation} \label{eq:rw_bias}
\P[Bin(c_2 \log n, 0.99) < \frac{c_2}{2} \log n + 2 \log_2 n] < n^{-\beta-1},
\end{equation}
Let $c_3 = k c_2 + 4 c_1$.  

We now describe an insertion step. 
Let $S$ denote a currently optimally sorted set. We will partition $S$ into 
consecutive intervals of length between $c_3 \log n$ and $2 c_3 \log n$ denoted
$I_1,\ldots,I_{t}$. We will use the notation $I_i'$ 
for the sub-interval of $I_i = [s,t]$ defined by 
$I_i' = [s+2 c_1 \log n,t-2 c_1 \log n]$. 
We say that a newly inserted element $a_j$ {\em belongs} 
to one of the interval $I_i$ 
if one of the two closest elements to it in the original order belongs to 
$I_i$. Note that $a_j$ can belong to at most two intervals. An element in $S$ belongs to $I_i$ iff it is one of the elements in $I_i$.
Note furthermore 
that if $a_j$ belongs to the interval $I_i$ then its optimal insertion location 
is determined up to $2 (k c_2 + 6 c_1) \log n$. Similarly, if we know 
it belongs to one of two intervals then its optimal insertion location 
is determined up to $4 (k c_2 + 6 c_1) \log n$, therefore we can take
$c = 4 (k c_2 + 6 c_1) = O(\ga^{-4}(\be+1))$.

Note that by the choice of $c_1$ 
we may assume that all elements belonging to $I_i$ are 
smaller than all elements of $I_j'$ if $i < j$ in the true order. 
Similarly, all elements belonging to $I_j$ are larger than all elements of 
$I_j'$ if $j > i$.  
We define formally the interval $I_0 = I_0'$ to be an interval of elements that are smaller than all the items and the interval $I_{t+1} = I_{t+1}'$ 
to be an interval of elements that is bigger than all items.

We construct a binary search tree on the set $[1,t]$ labeled by 
sub-intervals of $[1,t]$ such that the root is labeled by $[1,t]$ and 
if a node is labeled by an interval $[s_1,s_2]$ with 
$s_2 - s_1 > 1$ then its two children 
are labeled by $[s_1,s']$ and $[s',s_2]$, where $s'$ is chosen so that the 
length of the two intervals is the same up to $\pm 1$. Note that the two 
sub-interval overlap at $s'$. This branching process terminates at intervals 
of the form $[s,s+1]$. Each such node will have a path of descendants of length 
$c_2 \log n$ all labeled by $[s,s+1]$. 

We will use a variant of binary insertion closely related to 
the algorithm described in Section 3 of~\cite{FPRU:90}. The algorithm 
will run for $c_2 \log n$ steps starting at the root of the tree. 
At each step the algorithm will proceed from a node of the tree to either 
one of the two children of the node or to the parent of that node. 

Suppose that the algorithm is at the node labeled by $[s_1,s_2]$ and 
$s_2 - s_1 > 1$. 
The algorithm will first take $k$ elements from $I_{s_1-1}'$ that have not been 
explored before and will check that the current item is greater than the 
majority of them. Similarly, it will make a comparison with $k$ elements from 
$I_{s_2+1}'$. If either test fails it would backtrack to the 
parent of the current node. Note that if the test fails then it is the case
that the element does not belong to $[s_1,s_2]$ except with probability 
$10^{-2}$. 

Otherwise, let $[s_1,s']$ and $[s',s_2]$ denote the two children of 
$[s_1,s_2]$. The algorithm will now perform a majority test against $k$ elements from $I_{s'}$ according to which it would choose one of the 
two sub-interval $[s_1,s']$ or $[s',s_2]$. Note again that a correct 
sub-interval is chosen except with probability at most $10^{-2}$ (note that 
in this case there may be two ``correct'' intervals). 

In the case where $s_2 = s_1+1$ we perform only the first test. If it fails 
we move to the parent of the node. It it succeeds, we move to the single child. 
Again, note that we will move toward the leaf if the interval is correct 
with probability at least $0.99$. Similarly, we will move away from the leaf 
if the interval is incorrect with probability at least $0.99$. 

Overall, the analysis shows that at each step we move toward a leaf including 
the correct interval with probability at least $0.99$. From~\eqref{eq:rw_bias} it follows that with probability at least $1-n^{-\beta-1}$ after $c_2 \log n$ steps 
the label of the current node will be $[s,s+1]$ where the inserted element 
belongs to either $I_s$ or $I_{s+1}$. Thus the total number 
of queries is bounded by $3 k c_2 \log n$ and we can take
$C=3 k c_2 = O(\ga^{-2}(\be+1))$. This concluded the proof. 
\end{proof}

\newpage 

\bibliographystyle{abbrv}
\bibliography{all}

\end{document}